\documentclass[11pt]{article}
\usepackage{amsmath,amsthm,amsfonts}

\usepackage[letterpaper,top=1in,bottom=1in,left=1in,right=1in]{geometry}

\usepackage{amsmath,amsthm,amscd,amsfonts,amssymb,amstexalg}
\usepackage{verbatim}
\usepackage[classfont=bold,langfont=roman,funcfont=typewriter,full]{complexity}

\usepackage{mathrsfs} 

\DeclareMathOperator*{\argmax}{\arg\,\max}

\DeclareMathOperator{\AND}{\mathrm{AND}}

\newcommand{\zeroone}{0$-$1}

\newcommand{\SETH}{\mathbf{SETH}}

\newcommand{\reals}{\mathbb{R}}

\theoremstyle{plain}
\newtheorem{theorem}{Theorem}[section]
\newtheorem{lemma}[theorem]{Lemma}

\theoremstyle{definition}
\newtheorem{definition}{Definition}[section]

\newclass{\OPP}{OPP}
\newclass{\OP}{OP}
\newclass{\BPEXP}{BPEXP}

\newlang{\linspckt}{\textsc{Linear-size Series-Parallel Circuits}}
\newlang{\sat}{\textsc{Sat}}
\newlang{\wsp}{\textsc{Workflow Satisfiability Problem}}
\newlang{\formulasat}{\textsc{Formula Sat}}
\newlang{\subsetsum}{\textsc{Subset Sum}}
\newlang{\seth}{\mathbf{SETH}}
\newlang{\maxtwosat}{\textsc{Max-2-Sat}}
\newlang{\smt}{\textsc{smt}}
\newlang{\pit}{\textsc{pit}}
\newlang{\CAPP}{CAPP}
\newlang{\perm}{\textsc{Permanent}}
\newlang{\satisfiability}{\textsc{Satisfiability}}
\newlang{\uniquesat}{\textsc{unique-sat}}
\newcommand{\ksat}[1][k]{\lang{\textsc{$#1$-sat}}}

\newlang{\dkcsp}{\textsc{$(d,k)$-csp}}

\newlang{\cnf}{\textsc{cnf}}
\newlang{\maxsat}{\textsc{maxsat}}
\newlang{\dnf}{\textsc{dnf}}

\newlang{\cnfsat}{\textsc{cnfsat}}
\newlang{\cktsat}{\textsc{Circuit Sat}}
\newlang{\lincktsat}{\textsc{Linear-size Circuit Sat}}
\newlang{\turingsat}{\textsc{Turing Sat}}
\newlang{\ind}{\textsc{Independent Set}}
\newlang{\chromaticnumber}{\textsc{Chromatic Number}}
\newlang{\maxind}{\textsc{Maximum Independent Set}}
\newlang{\subiso}{\textsc{Subgraph Isomorphism}}
\newlang{\hamp}{\textsc{Hamiltonian Path}}
\newlang{\parityfunction}{\textsc{Parity}}
\newlang{\hamc}{\textsc{Hamiltonian Cycle}}
\newlang{\clique}{\textsc{Clique}}
\newlang{\colorability}{\textsc{Colorability}}
\newlang{\listcoloring}{\textsc{List Coloring}}
\newlang{\dominatingset}{\textsc{Dominating Set}}
\newlang{\feedbackvertexset}{\textsc{Feedback Vertex Set}}
\newlang{\longestpath}{\textsc{Longest Path}}
\newlang{\kcolorability}{\textsc{$k$-Colorability}}
\newlang{\threecolorability}{\textsc{$3$-Colorability}}
\newlang{\fourcolorability}{\textsc{$4$-Colorability}}
\newlang{\vertexcover}{\textsc{Vertex Cover}}
\newlang{\ksetcover}{\textsc{$k$-Set Cover}}

\title{0-1 Integer Linear Programming with a Linear Number of Constraints}
\author{{\large\sc Russell Impagliazzo}\thanks{This research is supported by NSF
grant CCF-1213151 from the
Division  of Computing and Communication Foundations.
Any opinions,
findings and conclusions or
recommendations expressed in this material are those
of the authors and do
not necessarily reflect the
views of the National Science Foundation.}
\\
\and {\large\sc Shachar Lovett}\thanks{Supported by NSF CAREER award 1350481.} \\
\and {\large\sc Ramamohan Paturi}\footnotemark[1] \\
\and {\large\sc Stefan Schneider}\footnotemark[1] \\
\\
Department of Computer Science and Engineering\\
University of California, San Diego\\La Jolla, CA 92093-0404, USA\\
E-Mail: \{russell, slovett, paturi, stschnei\}@cs.ucsd.edu
}
\date{January 2014}
\begin{document}

\begin{titlepage}
	\maketitle

We give an exact algorithm for the $\zeroone$ Integer Linear
Programming problem with a linear number of constraints that improves
over exhaustive search by an exponential factor. Specifically, our
algorithm runs in time $2^{(1-\text{poly}(1/c))n}$ where $n$ is the
number of variables and $cn$ is the number of constraints. The key
idea for the algorithm is a reduction to the \emph{Vector Domination
  problem} and a new algorithm for that subproblem.

	\thispagestyle{empty}
	\clearpage
\end{titlepage}

\section{Introduction}

A large number of real world problems from social sciences,
economics, logistics and other areas are very naturally expressed as
integer programs. Various variants of Integer Programming have been
studied, such as bounds on the solution vector, pure or mixed integer
programs, and linear, nonlinear or even nonconvex constraints, as well
as a number of other restrictions on the constraints. Most forms of
Integer Programming are $\NP$-hard, with some variants in $\P$ (such
as linear, totally unimodular constraints) and the variant with
nonconvex constraints and unbounded solution vector is even
undecidable.  For a survey on the large body of work for solving
various variants of Integer Programming exactly or approximately, we
refer to the survey by Genova and Guliashki
\cite{GenovaGuliashki_2011}.

In this paper we concentrate on the special case of $\zeroone$ Integer
Linear Programming (ILP). Given $n$ Boolean variables and $m$ linear
constraints, the problem is to find an assignment of either $0$ or $1$
to the variables such that all constraints are satisfied. For this
special case, we omit the objective function to be optimized and only
consider the problem of deciding if a set of constraints is
feasible. Since any objective value can be at most exponentially large
in the input size, binary search can reduce the optimization problem
to the feasibility problem with only polynomial overhead.

The problem can trivially be solved in time $O(2^n \text{poly}(n,m))$
using exhaustive search. On the other hand, an algorithm to solve the
problem in time $O\left(2^{(1-s)n}\right)$ for $s >0$ when $m$ is
superlinear in $n$ would contradict the Strong Exponential Time
Hypothesis ($\SETH$) \cite{ImpagliazzoPaturi_1999_jcss}, which says
that for every $s>0$ there is a $k$ such that $\ksat$ cannot be solved
in time  $O\left(2^{(1-s)n}\right)$. $\cnfsat$ on $m$ clauses is a
special case of $\zeroone$ ILP on $m$ constraints and the
Sparsification Lemma \cite{ImpagliazzoPaturiZane_1998_jcss} allows us to
reduce $\ksat$ to $\cnfsat$ on $c(k)n$ many clauses. Short of
proving or disproving $\SETH$, we can then ask the following question:
Given a $\zeroone$ integer linear program on $n$ variables and a
linear number of constraints, is it possible to decide feasibility
faster than exhaustive search? We answer this question
affirmatively. In particular, we prove the following.

\begin{theorem}
  \label{thmmain}
  Given a $\zeroone$ integer program on $n$ variables and $m = cn$
  constraints, it is possible to decide the feasibility of the program
  in time $2^{(1-s(c))n}$ where $s(c) =
  \Omega\left(\text{poly}\left(\frac1c\right)\right)$. 
\end{theorem}

For the special case of $\cnf$ satisfiability, Schuler
\cite{Schuler_2005_jalg} gives an algorithm that improves over
exhaustive search by an exponential factor if the number of clauses is
linear. With savings (the $s(c)$ in above theorem) that are inverse
polylogarithmic in $c$, as opposed to polynomial, Schuler's algorithm
runs considerably faster.

Recently, Williams \cite{Williams_2014_arxive} gave an algorithm for
$\zeroone$ ILP that improves over exhaustive search even for a
polynomial number of constraints. The algorithm runs in time
$2^{(1-s)n}$. For $s = \frac1{\text{polylog}(m)}$. Since $s$ is
subconstant, even for linear $m$, Williams' result is not directly
comparable to our result. Also note that a subconstant $s$ does not
contradict $\SETH$.

Our result is a follow-up to earlier work of a subset of the authors
\cite{ImpagliazzoPaturiSchneider_2013_focs}, where we considered the
more general class of depth two threshold circuits. Depth two
threshold circuits generalize the problem from a conjunction of linear
constraints to a threshold function of linear constraints. However,
the algorithm for depth two threshold circuits we gave only improves
over exhaustive search by an exponential factor if the number of
nonzero coefficients (across all constraints) is linear in the number
of inputs. Furthermore, the savings of our algorithm for depth two
threshold circuit was exponential in $\frac1c$.

The key idea of our algorithm is to reduce the problem to the
\emph{Vector Domination problem}, the problem of finding a pair of
vectors such that one vector dominates the other in every
coordinate. As such, the main idea stays the same as in our previous
work. The main technical contribution of the current paper is an
algorithm for the Vector Domination problem that improves over the
trivial algorithm for vectors of dimension $O(\log N)$ (where $N$ is
the number of vectors), whereas earlier algorithms, such as the one
used for our previous result, only worked for dimensions up to $\delta
\log N$, where $\delta$ is a sufficiently small number. Whereas
algorithms for small dimensions have been discussed before, in
particular Bentley \cite{Bentley_1980} and Chan \cite{Chan_2005}, we
believe that we give the first algorithm improving over exhaustive
search for all dimensions possible without refuting $\seth$.


\section{ILP and the Vector Domination problem}
\label{Domination}
In this section we give a reduction from $\zeroone$ Integer Linear
Programming to the \emph{Vector Domination} problem.

In all definitions below, $u \geq v$ for two vectors $u$ and $v$ of
the same dimensions is used to denote an element-wise comparison.

\begin{definition}
  Given a matrix $M \in \reals^{n \times m}$ and a vector $r \in
  \reals^{m}$, the \emph{$\zeroone$ Integer Linear Programming}
  problem on $n$ variables and $m$ constraints is to find a vector $x
  \in \{0,1\}^n$ such that $Mx \geq r$.
\end{definition}

\begin{definition}
  Given two sets of $d$-dimensional real vectors $A$ and $B$, the
  \emph{Vector Domination Problem} is the problem of finding two
  vectors $u \in A$ and $v\in B$ such that $u \geq v$.
\end{definition}

Both problems have a trivial exhaustive search algorithm. For
$\zeroone$ ILP, that algorithm runs in time $O\left( 2^n
  \text{poly}(n,m)\right)$, while the trivial algorithm for the Vector
Domination problem runs in time $O(|A| |B| d)$.  

We can reduce $\zeroone$ ILP to the Vector Domination problem, which
shows that if the Vector Domination problem allows for an algorithm
faster than exhaustive search, then so does $\zeroone$ ILP. The
reduction was introduced by Williams \cite{Williams_2005_tcs} for the
special case of Boolean vectors (he called the problem the Cooperative
Subset Query problem).

\begin{lemma}
  \label{lemilp}
  Suppose there is an algorithm for the Vector Domination problem for
  $|A| = |B| = N$ and $d = 2c \log N$ running in time
  $O\left(N^{2-s(c)}\right)$. Then a $\zeroone$ Integer Linear program
  on $n$ variables and $m = cn$ constraints can be solved in time
  $O\left(2^{(1 - s(c)/2) n}\right)$.
\end{lemma}
\begin{proof}
  Let $M$ and $r$ be the matrix and vector respectively of the
  $\zeroone$ integer linear program.  Separate the variable set into
  two sets $S_1$ and $S_2$ of equal size. For every assignment
  to the variables in $S_1$ and $S_2$, we assign a $c n$-dimensional vector where
  every dimension corresponds to a constraint. Let $\alpha$ be an
  assignment to $S_1$ and let $a \in \reals^{c n}$ be the
  vector with $a_j = \sum_{i\in S_1} M_{i,j} \alpha(x_i)$ and let
  $A$ be the set of $2^{n/2}$ such vectors. For an assignment $\beta$
  to $S_2$, let $b$ be the vector with $b_j = r_j - \sum_{xi \in S_2}
  M_{i,j}\beta(x_i)$ and let $B$ be the set of all such vectors $b$.

  An assignment to all variables corresponds to an assignment to $S_1$
  and an assignment to $S_2$, and hence to a pair $a \in A$ and $b \in
  B$. The pair satisfies all inequalities if and only if $a$ dominates
  $b$. We have $|A| = |B| = N = 2^{n/2}$ and $d = cn = 2c
  \log(N)$. Hence the Vector Domination problem, and therefore the
  $\zeroone$ Integer Linear Program, can be solved in time
  $O\left(N^{2-s(c)}\right) = O\left(2^{(1 - s(c)/2) n}\right)$.
\end{proof}

Note that the reduction can be trivially adapted to any ILP variant
where the variables can take values from a constant size set.


\section{An Algorithm for the Vector Domination Problem}
The algorithm is a divide and conquer algorithm, splitting the two
sets $A$ and $B$ into sets $A^+$,$A^-$, $B^+$ and $B^-$ depending on
if the first coordinate is larger or smaller than some value $a$. For
our purposes, we choose $a$ as a weighted median that is computed
across both sets $A$ and $B$.

The weighted median of a collection of real numbers with associated
real weights is a number such the both the total weight of all numbers
smaller than the median and the total weight of all numbers larger
than the median are at most half of the total weight. Note that the
weighted median can be computed in linear time using minor
modifications to the standard algorithms to compute the unweighted
median.

\begin{lemma}
  \label{lemalgo}
  Let $A, B \subseteq \reals^d$ with $|A| = |B| = N$ and $d = c \log
  N$. There is an algorithm for the Vector Domination problem that
  runs in time $O\left(N^{2-s(c)}\right)$ , where $s(c)$ is polynomial
  in $\frac1c$.

\end{lemma}
\begin{proof}
  We assume w.l.o.g. that $c \geq 4$. Let $\varepsilon = \frac1{c^{15}}$,
  $\gamma = \frac{1}{15 \log{c}}$ and $t = \gamma \log N$. 

  Furthermore, let $a$ be the weighted median of the first coordinates
  of $A \cup B$, where all numbers from $A$ have weight $|B|$ and all
  number from $B$ have weight $|A|$. Further let $A^+ \subseteq A$
  consist of all vectors where the first coordinate is larger than $a$
  and $A^-$ consist of all vectors where the first coordinate is
  smaller than $a$. Similarly, split $B$ into two sets $B^+$ and
  $B^-$. For vectors where the first coordinate is exactly $a$, it is
  sufficient to split the vectors evenly between the two possible
  sets, as long as we do not at the same time add vectors to $A^-$ and
  $B^+$. This rounding is equivalent to adding a small positive noise
  to all vectors in $A$ and a small negative noise to all vectors in
  $B$.

  Now a vector $u \in A$ can only
  dominate a vector $v \in B$ in one of three cases:
  \begin{enumerate}
  \item $u \in A^+$ and $v\in B^+$
  \item $u \in A^-$ and $v\in B^-$
  \item $u \in A^+$ and $v\in B^-$ 
  \end{enumerate}

  Also, in the third case any vector in $A^+$ dominates any vector in
  $B^-$ on the first coordinate, hence our recursive algorithm can
  recurse on $d-1$ dimensions.

  Since we split at a weighted median where both $A$ and $B$ have the
  same total weight, we have $\frac{|A^-|}{|A|} = \frac{|B^+|}{|B|} =:
  \varepsilon'$. We distinguish two cases, a \emph{balanced} case
  where $\varepsilon' \geq \varepsilon$ and an \emph{unbalanced} case
  where $\varepsilon' < \varepsilon$. In both the balanced and the
  unbalanced cases, we recurse on all three subproblems. In the
  balanced case, we additionally decrement $t$. Once $t = 0$ we solve
  the Vector Domination problem on the remaining vectors by exhaustive
  search.

  To bound the runtime of above algorithm, we consider the recursion
  tree. We first bound the time spent on exhaustive search in leaves
  where $t= 0$. Each of the $N^2$ possible pairs of dominating vectors
  appears in at most one of the subcases. Furthermore, in the balanced
  case there are at least $\varepsilon^2 N^2$ pairs where one vector
  is in $A^-$ and the other is in $B^+$ that are not considered in any
  subcase. Since there are $t$ balanced cases on the path from the
  root to any leaf with $t=0$, the total time spent on exhaustive
  search is bounded by
  \begin{equation*}
    (1- \varepsilon^2)^t N^2 \leq e^{-\varepsilon^2 \gamma \log N} N^2
    = N^{2-\log(e) \varepsilon^2 \gamma} = N^{2-\frac{\log(e)}{15
        c^{30} \log(c)}} = N^{2-\text{poly}(1/c)}
  \end{equation*}
  For $c=4$, we get constant savings of $2^{-66}$.

  To bound the size of the recursion tree we bound the number of
  possible paths from the root to any leaf. On any path, there are at
  most $d$ steps that decrease the dimension. Furthermore, there are
  at most $t$ balanced cases on any path. For the unbalanced case, in
  both subproblems 1 and 2, where the dimension does not decrease, the
  number of pairs of vertices decreases by a factor of at most
  $\varepsilon$, hence this can happen at most
  $\frac{\log(N^2)}{\log(1/\varepsilon)}$ times along any path. Let $r
  = t + \frac{\log(N^2)}{\log(1/\varepsilon)} = \left(\gamma +
    \frac2{\log(1/\varepsilon)}\right) \log N = \frac1{5 \log(c)}
  \log(N)$. Taking into account that the time spent on computing the
  median in every node is linear in $|A| + |B| = O(N)$, the total time
  is bounded by
  \begin{align*}
    O(N) \binom{d+r}{r} 2^r &\leq
    O(N) \left(e\left(1 + 5 c
        \log(c)\right)\right)^{\frac1{5\log(c)} \log N}
    N^{\frac1{5\log(c)}} \\ &= N^{1+\log\left(e\left(1 + 5 c
            \log(c)\right)\right)\frac1{5\log(c)} +
        \frac1{5\log(c)}}
  \end{align*}

  The bound is monotonely decreasing for $c\geq 4$ and at $c=4$ the
  bound is $O\left(N^{1.781}\right)$.

  The overall runtime is therefore always dominated by the time spent
  on exhaustive search in the leaves, which gives us the claim
  immediately.
\end{proof}

The algorithm can be extended for the case where $|A| \neq
|B|$. If $|A| = O(|B|)$, above analysis immediately gives us a runtime
of $(|A| |B|)^{1-s(c)/2}$. If $|A| \gg |B|$, we partition $A$ into $\frac{|A|}{|B|}$ subsets of
size $|B|$ each and solve the Vector Domination problem for each
subset individually. The whole algorithm then runs in time
$|B|^{2-s(c)} \frac{|A|}{|B|} = |A| |B|^{1-s(c)}$. 

Theorem \ref{thmmain} then follows directly from Lemmas
\ref{lemilp} and \ref{lemalgo}.



\section{Conclusions}
We give a first algorithm for $\zeroone$ integer programs on $n$
variables and $cn$ constraints that improves over exhaustive search by
a factor $2^{\text{poly(1/c)}n}$. The result does generalize to ILP
where the variables are constrained by any finite set of values, not
just Boolean variables.

Under the Strong Exponential Time Hypothesis, this is qualitatively
optimal in the sense that we can only expect exponential improvement
over exhaustive search if the number of constraints is
linear. However, for the special case of formulas in conjunctive
normal form the best algorithms achieve savings that are
polylogarithmic in $\frac1c$ \cite{Schuler_2005_jalg}. It is open if
we can get the same for $\zeroone$ ILP.

This result is not comparable to our earlier work
\cite{ImpagliazzoPaturiSchneider_2013_focs}. While $\zeroone$ ILP is a
special case of depth two threshold circuits, our result requires a
linear number of constraints instead of a linear number of wires. It
is still open if we can find an algorithm for general depth two
threshold circuits that works for a linear number of gates
(i.e. constraints).

Lastly, there are countless problems that reduce to Integer
Programming in a natural way. Some of these applications might benefit
from this result. Alternatively, there might be problems that reduce
directly to the Vector Domination problem and benefit from the
corresponding subroutine.

\bibliographystyle{plain}
\bibliography{TeX/bib/complexity.bib}

\begin{thebibliography}{1}

\bibitem{Bentley_1980}
Jon~Louis Bentley.
\newblock Multidimensional divide-and-conquer.
\newblock {\em Communications of the ACM}, 23(4):214--229, 1980.

\bibitem{Chan_2005}
Timothy~M Chan.
\newblock {All-pairs shortest paths with real weights in O (n 3/log n) time}.
\newblock In {\em Algorithms and Data Structures}, pages 318--324. Springer,
  2005.

\bibitem{GenovaGuliashki_2011}
Krasimira Genova and Vassil Guliashki.
\newblock Linear integer programming methods and approaches--a survey.
\newblock {\em Cybernetics And Information Technologies}, 11(1), 2011.

\bibitem{ImpagliazzoPaturi_1999_jcss}
R.~Impagliazzo and R.~Paturi.
\newblock The complexity of $\ksat$.
\newblock {\em Journal of Computer and Systems Sciences}, 62(2):367--375, March
  2001.
\newblock Preliminary version in {\em 14th Annual IEEE Conference on
  Computational Complexity}, pages 237--240, 1999.

\bibitem{ImpagliazzoPaturiZane_1998_jcss}
R.~Impagliazzo, R.~Paturi, and F.~Zane.
\newblock Which problems have strongly exponential complexity?
\newblock {\em Journal of Computer and System Sciences}, 63:512--530, 1998.
\newblock Preliminary version in {\em 39th Annual IEEE Symposium on Foundations
  of Computer Science}, pp 653-662, 1998.

\bibitem{ImpagliazzoPaturiSchneider_2013_focs}
Russell Impagliazzo, Ramamohan Paturi, and Stefan Schneider.
\newblock A satisfiability algorithm for sparse depth two threshold circuits.
\newblock In {\em Proceedings of the 54rd Annual IEEE Symposium on Foundations
  of Computer Science}, 2013.

\bibitem{Schuler_2005_jalg}
R.~Schuler.
\newblock An algorithm for the satisfiability problem of formulas in
  conjunctive normal form.
\newblock {\em Journal of Algorithms}, 54(1):40--44, 2005.

\bibitem{Williams_2014_arxive}
R.~{Williams}.
\newblock {New algorithms and lower bounds for circuits with linear threshold
  gates}.
\newblock {\em ArXiv e-prints}, January 2014.

\bibitem{Williams_2005_tcs}
Ryan Williams.
\newblock A new algorithm for optimal 2-constraint satisfaction and its
  implications.
\newblock {\em Theoretical Computer Science}, 348:357--365, 2005.

\end{thebibliography}
\end{document}